\documentclass[11pt,a4paper]{article}
\usepackage{lmodern,fullpage}
\usepackage{amsmath}
\usepackage{amsthm}
\usepackage{amssymb}
\usepackage{amsfonts}
\usepackage{mathrsfs}
\usepackage{mathtools}
\usepackage[utf8]{inputenc}
\usepackage{cmap} %enable copy-paste
\usepackage{outlines}
\usepackage{graphicx}
\usepackage{float}
\usepackage{cite}
\usepackage[bookmarks=true, hypertexnames=true]{hyperref} 

\newcommand{\nonre}{\mathbb{R}_{+}}

\newcommand{\eps}{\varepsilon}
\newcommand{\oeps}{(1+\eps)}
\newcommand{\ieps}{(1/\eps)}

\newcommand{\cm}{\mathtt{C}_{max}}
\newcommand{\cmm}{\mathcal{C}_{max}^{nice}}
\newcommand{\bcon}{\textsf{time constraint}}
\newcommand{\nbcon}{\textsf{modified time constraint}}
\newcommand{\epsb}{$\eps-$block}
\newcommand{\uepsb}{universal \epsb}
\newcommand{\onei}{$1-$interval}
\newcommand{\nonei}{$1'-$interval}
\newcommand{\iepsb}{$(1/\eps)-$block}

\newcommand{\opt}{\textsf{opt}}
\newcommand{\sol}{\textsf{sol}}
\newcommand{\copt}{C_{\opt}}

\newcommand{\sch}{\sigma}
\newcommand{\tsch}{\tilde{\sigma}}
\newcommand{\bcs}{\texttt{STS}}
\newcommand{\la}{\mathtt{L}}
\newcommand{\sm}{\mathtt{S}}

\newcommand{\cons}{\mathtt{T_S}}
\newcommand{\conls}{\mathtt{CN_L}}
\newcommand{\conl}{\mathtt{T_L}}
\newcommand{\tj}{\mathtt{TJ}}
\newcommand{\sj}{\mathtt{SJ}}
\newcommand{\mj}{\mathtt{MJ}}
\newcommand{\lj}{\mathtt{LJ}}
\newcommand{\conts}{\mathtt{T}}
\newcommand{\cont}{\mathsf{CN}}
\newcommand{\confs}{\mathtt{C}}
\newcommand{\conf}{\mathsf{CF}}

\newtheorem{lemma}{Lemma}
\newtheorem{theorem}{Theorem}
\newtheorem{corollary}{Corollary}

\title{EPTAS for parallel identical machine scheduling with time restrictions}
\author{G. Jaykrishnan\thanks{Faculty of Industrial Engineering and Management, The Technion, Haifa, Israel. {\bf Email}: jaykrishnang@hotmail.com.}  \and Asaf Levin\thanks{ 
		Faculty of Industrial Engineering and Management, The Technion, Haifa, Israel. {\bf Email: }{levinas@ie.technion.ac.il.} Partially supported by ISF - Israeli Science Foundation grant number 308/18.}}
\date{ }
\begin{document}
	\maketitle

\begin{abstract}
We consider the non-preemptive scheduling problem on identical machines where there is a parameter $B$ and each machine in every unit length time interval can process up to $B$ different jobs.  The goal function we consider is the makespan minimization and we develop an EPTAS for this problem.  Prior to our work a PTAS was known only for the case of one machine and constant values of $B$, and even the case of non-constant values of $B$ on one machine was not known to admit a PTAS.  
\end{abstract}
	
	\section{Introduction}
		The problem considered is the non-preemptive scheduling of $n$ jobs on $m$ parallel identical machines such that each unit time interval cannot intersect more than $B$ jobs on each machine. That is, for a time interval $[\alpha, \alpha+1)$, for $\alpha\in\nonre$ and for a given machine $i$, the number of jobs being (partially or completely) processed on $i$ in this time interval cannot be more than $B$. This constraint is referred to as \bcon{} and the problem is denoted as \bcs\ (Scheduling with Time Restrictions on identical machines). 
Formally, the problem is defined as follows. The input consists of $n$ independent jobs and $m$ identical parallel machines, each job $j$ has a processing time $p_j\in\nonre{}$ also known as the size of $j$, and finally there is a positive integer $B \geq 2$.   A {\it non-preemptive schedule} assigns every job $j$ a pair consisting of a machine and an interval along the time horizon such that the length of the time interval assigned to $j$ is exactly $p_j$ and two such time intervals assigned to a pair of jobs $j,j'$ both of which assigned to a common machine, do not intersect in an inner-point.  The makespan, that is the maximum completion time of a job in such non-preemptive schedule is the goal function that we aim to minimize, and our problem, \bcs\ is to find a non-preemptive schedule subject to \bcon\ saying that in every time interval $I$ of length $1$ and every machine $i$, the number of jobs assigned to $i$ that the schedule assigns to time intervals intersecting $I$ is at most $B$.

A $\rho$-approximation algorithm for a minimization problem is a polynomial time algorithm that always finds a feasible solution of cost at most $\rho$ times the cost of an optimal solution. A polynomial time approximation scheme (PTAS) for a given problem is a family of approximation algorithms such that the family has a $(1+\varepsilon)$-approximation algorithm for any $\varepsilon>0$. An efficient polynomial time approximation scheme (EPTAS) \cite{Cesati97,Downey99,Flum2006} is a PTAS whose time complexity is upper bounded by the form $f(\frac{1}{\varepsilon}) \cdot poly(n)$ where $f$ is some computable (not necessarily polynomial) function and $poly(n)$ is a polynomial of the length of the (binary) encoding of the input. A fully polynomial time approximation scheme (FPTAS) is defined like an EPTAS, with the added requirement that $f$ must be upper bounded by a polynomial in $\frac 1{\varepsilon}$.

The \bcon\  is introduced in	\cite{Braun2014} who consider the special case of our problem  on a single machine. They prove that the the problem is NP-hard when the value of $B$ is a variable and performed worst-case analysis of List Scheduling (LS) algorithm for this case of our problem. Later studies considering \bcon\ were considering only the special case of our problem with one machine (i.e., $m=1$).
 \cite{Benmansour2014} presented a time-indexed mixed integer linear programming formulation for it  (see also \cite{Benmansour2019} for other mixed-integer linear programming formulations of the problem). \cite{Braun2016} improved the analysis for LS and analyzed the worst-case behavior of Longest Processing Time schedules.  A better permutation of jobs that improved the bounds is derived in \cite{Zhang2017}. Most relevant to our work, \cite{Zhang2018} shows that the problem is NP-hard even when $B=2$ and exhibit a PTAS for the problem with any constant $B\geq2$.  We mention the fact that the last PTAS of \cite{Zhang2018} is not an EPTAS and does not lead to a PTAS for the special case of \bcs\ with one machine and a non-constant value of $B$.   We will show here an EPTAS for \bcs\ even for an arbitrary (and no-constant) number of identical machines, and as a by-product improve the approximation guarantee on the problem for one machine.

Minimizing the makespan on identical machines is a special case of \bcs\ that correspond to the case $B=n$.  Hochbaum and Shmoys \cite{Hochbaum87} showed an EPTAS (see also \cite[Chapter~9]{hochbaum97} and \cite{Alon98}) for the makespan minimization objective for this setting. Jansen et al. \cite{JKV16} improved the time complexity of this EPTAS (also for related problems).  Note that minimizing the makespan on identical machines is known to be NP-hard in the strong sense, and thus unless $P=NP$, there is no FPTAS for it and so also for problem \bcs\ that we study here.  Therefore, EPTAS is the fastest class of approximation schemes that can be obtained to our problem.  

Last, the special case where the algorithm knows in advance that the optimal makespan is smaller than $1$ is equivalent to the scheduling problem on identical machines with the constraint saying that every machine can be assigned at most $B$ jobs.  This case known as scheduling with cardinality constraint and its generalization saying that every machine may have different capacity admit EPTAS (on identical machines)
\cite{LJ+16}.  This EPTAS will be used here for one simple case.	See \cite{DM01,He03,DIMM06,KK13,LJ+16,Ka21} for further studies of scheduling with cardinality constraint.

\paragraph{Outline of our scheme and methods.}  Our scheme starts by rounding the job sizes, guessing the makespan value for using the dual approximation approach of \cite{Hochbaum87,Hochbaum88}, and characterizing a sub-class of solutions containing at least one nearly optimal solution, see Section \ref{sec:round}.   Next, we formulate a mixed-integer linear program (MILP) in Section \ref{sec:MILP} based on a hierarchy of configurations.  The use of integer linear programs in fixed dimension and of mixed-integer linear programs with fixed number of integral decision variables is common in the design of EPTAS's for scheduling and load balancing problems, see e.g. \cite{Hochbaum87,Alon98,Jansen10,epstein2014,EL14c,LJ+16,JKV16,Jansen2019,Kones2019}.  Here, we will use the term configuration to encode the schedule of one machine, and we use the term container to refer to a sub-schedule of one configuration restricted to a portion of the time horizon where the job set assigned to it are relatively short.  This will allow us to have a constant number of configurations and a constant number of containers, but the job assignment to machines based on the counters of these terms is not well-defined.  Such notion of hierarchy of configurations was used before, see e.g. \cite{JKMR19}. However, here we have a polynomial number of fractional decision variables and constraints so unlike \cite{JKMR19} we will need to allow most of them to be fractional before applying Lenstra's algorithm \cite{lenstra1983,kannan1983}.  The rounding of a MILP solution into a feasible schedule is the main technical part of our scheme.  The most non-standard step in this rounding phase is the rounding of the tiny jobs based on a non-standard use of the integrality of polytopes defined using an interval matrix and integer right hand side followed by a step analyzed in  \cite{LJ+16}.  While in Section \ref{sec:MILP} we show that there is a feasible solution for the MILP for a correct guessed value  of the makespan, the rounding of the optimal solution of the MILP is presented and analyzed in Section \ref{sec:round-MILP}.

\section{Rounding, guessing, and a characterization of near optimal solutions\label{sec:round}}
\paragraph{Preliminaries: definitions and notation.} Let $\eps>0$ be such that $1/\eps$ is an integer larger than $3$.  A unit time interval is denoted by \textit{\onei{}}. A time interval of length $1+\eps$ that starts at integer multiple of $\eps$ is defined as a \textit{\nonei{}}. A \textit{schedule} assigns jobs to machines and starting times to jobs. Let $J$ be the set of jobs and $M$ the set of machines. The {\em load} of a machine is the finishing time of the last job assigned to the machine.  Note that schedules may have idle time so the load of a machine is not necessarily the total size of jobs assigned to that machine. 

We define blocks of schedules in order to facilitate the definition of containers in the next section. A block of schedule of length $1/\eps$ of one specific machine $i$ is referred to as a \textit{\iepsb{}} and it is the set of jobs that are assigned to start, parts of jobs, and the amount of idle time in $1/\eps$ length of time of the schedule of machine $i$. That is, the set of jobs whose processing intervals on machine $i$ intersects (in a non-empty interval) the given time interval of length $1/\eps$. A block of schedule of one particular machine of length $\eps{}$ starting at integer multiple of $\eps$ is called an \textit{\epsb{}}, i.e., the set of jobs that are assigned to start, parts of jobs, and the amount of idle time assigned to be processed on this machine in the given $\eps{}$ length time interval that starts at an integer multiple of $\eps$ of the time horizon. 

We define a new (more strict as we show in the next lemma) constraint called \textit{\nbcon{}}. A schedule satisfies \nbcon{} if no more than $B$ jobs assigned to a common machine intersect with every \nonei{} of the schedule.

\begin{lemma}\label{BB'}
	If a schedule satisfies \nbcon\ then it also satisfies the \bcon{}.
\end{lemma}
\begin{proof}
	Every \onei{} that starts at an arbitrary point in time is contained within a \nonei{}. Since the number of jobs being processed in every \nonei{} on a given machine is at most $B$, every \onei{} of that machine contains at most $B$ jobs.
\end{proof} 

\paragraph{The rounding step.}	Round up the size of each job to the next integer power of $1+\eps$. Denote the rounded size of job $j$ by $p'_j$. That is,	$p'_j=\oeps^{\left\lceil \log_{\oeps} p_j \right\rceil}$. 
	Notice that,
$ p_j \leq p'_j \leq \oeps p_j$.

	Denote by $I'$ the rounded instance where the size of all the jobs are rounded as specified. Next, we analyze this rounding step using the standard argument of time stretching.
	
	\begin{lemma}
		Any schedule $\sch$ feasible to instance $I$, of makespan $C_{max}$, can be used to generate a schedule $\sch'$ which is feasible to the rounded instance $I'$ with makespan at most $\oeps C_{max}$. Any schedule $\sch'$ feasible to the rounded instance $I'$, of makespan $C'_{max}$, is feasible to the original instance $I$ with makespan at most $C'_{max}$.
	\end{lemma}
	\begin{proof}
		Consider a schedule $\sch$ feasible to instance $I$ of makespan $C_{max}$.   Consider one specific machine $i$ (we apply the process for each machine) and denote its load in $\sch$ by $L_i$. Convert $\sch$ to $\sch'$ by rounding up the sizes of all the jobs, without changing the ordering of the jobs assigned to $i$, and multiplying the starting time of every job by $1+\eps$. This process guarantees that the number of jobs in each \onei{} is not increased. Thus \bcon{} is satisfied by $\sch'$. The new load of $i$ is at most $(1+\eps)L_i$,
		and thus the new makespan satisfies
		$C'_{max} \leq \oeps C_{max}$ as required.
		
For the other direction, consider a schedule $\sch'$ which is feasible to the instance $I'$ of makespan $C'_{max}$ and one machine $i$. The jobs in $\sch'$ are rounded down to get instance $I$ but the starting time of the jobs remain fixed, i.e. same as in $\sch'$. Thus additional idle time is added as necessary during the conversion.  We denote by $\sch$ the resulting schedule of $I$. Since $\sch'$ satisfied the \bcon{}, the number of jobs in every \onei{} of $\sch'$ assigned to $i$ is at most $B$. Since the starting time of the jobs does not change, every \onei{} in $\sch$ assigned to $i$ also contains at most $B$ jobs. Thus, $\sch$ is feasible and the makespan of $\sch$ is at most $C'_{max}$.
	\end{proof}
	
With a slight abuse of notation we assume without loss of generality that the input $I$ is already rounded, that is, $I=I'$ and so $p_j$ will denote the {\em rounded size of job $j$}.	
	
\paragraph{Guessing step.}
We guess the makespan of the schedule rounded up to an integer power of $1+\eps$. An a-priori valid upper bound on the makespan is $n + \sum_{j\in J}p_j$. The lower bound on the makespan is
$\frac{\sum_{j\in J}p_j}{m}$.  Using these bounds will lead to a weakly polynomial upper bound on the number of possibilities for this guessed information.  Thus, we consider the following bound.  Let $i$ be a machine attaining the makespan in a fixed optimal solution and let $p_{\max}$ be the maximum size of a  job in $I'$.  Since the number of jobs assigned to $i$ is between $1$ and $n$, we conclude that the makespan is not smaller than $p_{\max}$ and not larger than $n(1+p_{\max})$.  If the guessed makespan is smaller than $1$, we will have a different procedure that does not use the guessed value of the makespan. Thus, the number of guessed values of at least $1$ is not larger than $\log_{1+\eps} (2n)$.  We conclude the following lemma.

\begin{lemma}
	The number of possible values for the guess of the makespan that are at least $1$ is at most $O(\log_{1+\eps} n)$, that is there are strongly polynomial number of possibilities.
\end{lemma}

The guess of the makespan is denoted by $\cm$. If $\cm<1$, then we apply the scheme of \cite{LJ+16} with a uniform capacity bound of $B$ on each machine, to obtain an EPTAS for our problem.  In what follows, we assume that $\cm\geq 1$ and observe that the number of such possibilities is at most $O(\log_{1+\eps} n)$ based on the last lemma.

\paragraph{Classification of jobs.}
Let $\theta = \max\{\eps\cm,\ieps\}$. A job $j$ is classified as
\textit{tiny} if $p_j \leq \eps^2$, \textit{small} if $\eps^2 < p_j \leq 1/\eps$, \textit{medium} if $1/\eps < p_j \leq \theta$, and \textit{large} if $\theta < p_j \leq \cm$. There can be instances where the interval of sizes of medium jobs and perhaps the interval of sizes of large jobs are empty. In particular when $\cm$ is at most $1/\eps$ there are no medium or large jobs. 

The number of distinct sizes of small jobs is at most $\log_{\oeps}\ieps^3 + 1$. Let $\sm$ denote the set of distinct sizes of small jobs. The number of distinct sizes of large jobs is at most $\log_{\oeps} \ieps + 1$. Let $\la$ be the set of distinct sizes of large jobs.  So $|\sm|$ and $|\la|$ are upper bounded by functions of $\eps$.  The number of distinct sizes of medium or tiny jobs cannot be upper bounded by a function of $\eps$ and these values will be upper bounded by polynomial (e.g., at most $n$ such distinct sizes).

Let $\tj$, $\sj$, $\mj$, $\lj$ be the set of tiny, small, medium, and large jobs, respectively. Any of the aforementioned job sets with a subscript $l$ denotes the subset of jobs of that set of size $l$.

\paragraph{Definition of nice solutions.}
A schedule is \textit{nice} if the schedule on each machine satisfies the following four conditions.
\begin{enumerate}
	\item The jobs in the schedule appear in the order of tiny or small jobs followed by medium jobs and then large jobs.

	\item In every \epsb{} of the schedule the idle time is at the beginning of the \epsb{} or at the completion time of the job that is being processed at the start of the \epsb{} and the jobs that start in the \epsb{} continue their processing continuously, in non-decreasing order of processing time.

	\item There is unit idle time between processing a pair of jobs if there is a difference of at least $1/\eps$ time between their starting times. 

	\item The schedule satisfies the \nbcon{}.
\end{enumerate}

We will use the set of nice solutions as a family of highly-structured solutions with the property that approximating the best solution among this set is sufficient for approximating the original problem.  Thus, we base our approach on the following lemma.

\begin{lemma}\label{nic_sol}
	Every feasible schedule $\sch$ of makespan $C_{\sch}$ can be converted to a nice schedule $\sch'$ of makespan at most $(1+17\eps) C_{\sch}$.
\end{lemma}
\begin{proof}
	Consider a feasible schedule $\sch$ of makespan $C_{\sch}$. Perform the following operations for each machine $i$. Let the load of machine $i$ be $L_i$, and the schedule for machine $i$ in $\sch$ be $\sch_i$. Feasibility ensures that the starting time of the jobs and the idle time in $\sch_i$ are such that the schedule satisfies the \bcon{} and the scheduling constraints (i.e., this is a non-preemptive schedule on machine $i$). Convert $\sch_i$ to a schedule with the same set of jobs on machine $i$ as follows.  
	\begin{outline}[enumerate]
		\1 \label{eps_idle_time} At integer points of time in $\sch_i$ add idle time of $2\eps$ as follows. At each integer point of time $t$ there can be one of the following two cases.
		
		\begin{itemize}
			\item When no job is being processed at $t$. Add idle time of $2\eps$ at $t$.
			\item When a job $j$ is being processed at $t$. Add idle time of length $2\eps$ at the completion time and the starting time of $j$.
		\end{itemize}
	We do this step simultaneously to all points along the schedule that were integer points in time before this step (i.e., in $\sch_i$). No matter how the starting time of the added idle time is modified the added idle time will contain an \epsb{} of idle time. This step can increase the length of the schedule by at most $4\eps$ for each integral point of time in $\sch_i$ so the load of $i$ will increase by at most $4\eps L_i$. 
		
		\1 \label{ieps_idle_time} At the starting time and completion time of medium and large jobs add idle time of length $2$. This step increases the load of $i$ by at most $4\eps L_i$ since a size of a medium/large job is at least $1/\eps$.
		
		\1 \label{med_lar_rem} For each medium and large job do the following. Remove the job from the schedule and then shift the remaining schedule backward to the starting time of the job. 
		After the end of the partial schedule obtained by removing all the medium and large jobs add a unit idle time and then place the medium and large jobs in non-decreasing order of the size. 
		
		This step basically pulls the medium and large jobs from the schedule and move them to the end of the schedule (of the same machine $i$).  The ordering of small jobs, tiny jobs, and idle time (until the first medium or large job) is kept as the outcome of step \ref{ieps_idle_time}. The unit length idle time before the first medium/large job adds at most $\eps L_i$ to the load of $i$.
		
		\1 \label{ieps_idle_time2} At integer multiples of $1/\eps$ in the schedule obtained from $\sch_i$ after Step \ref{med_lar_rem} add idle time as follows. At each integer multiple $t$ (simultaneously for all $t$) one of the following two cases hold.
		
		\begin{itemize}
			\item When no job is being processed at $t/\eps$, add $2$ units of idle time at $t/ \eps$. 
			
			\item When a job $j$ is being processed at $t/ \eps$, add idle time of length $2$ at the finishing time of $j$. 
		\end{itemize}
		No matter how the starting time of the added idle period will be modified in the future it will contain a \onei{} of idle time starting at an integer multiple of $\eps$. This step increases load of the machine by at most $8\eps L_i$ since the load of $i$ after the previous step was at most $(1+9\eps)L_i \leq 4L_i$.  
		
		\1 \label{rearr} For each $\eps-$block do the following. Calculate the total idle time in the \epsb{}. If there is a job that starts in an earlier \epsb{} and finishes in the current \epsb{}, then move this total idle time to the finishing time of that job, else to the beginning of the \epsb{}. Sort all the jobs starting in this \epsb{} in non-decreasing order of size. Now the machine operates continuously till the last job assigned to start in the \epsb{} is finished. 
	\end{outline}
	
	Let $\sch'$ be the schedule of machine $i$ obtained after all the above operations. After Step \ref{eps_idle_time}, $\sch$ is converted to schedule $\tsch$. Consider a \nonei{}, $O$, of machine $i$ in $\tsch$. There are two cases to deal with. First case is when there is idle time (added in Step 1) of at least $\eps{}$ in $O$. Then the set of jobs, complete and fractional, in $O$ used to belong to a \onei{} in $\sch$. Thus since this set of jobs in $O$ satisfied the \bcon{} in $\sch$, they do not violate the \nbcon{} in $\tsch$. Second case is when idle time (added in Step 1) in $O$ is of length strictly less than $\eps$. This happens because the unique job in $O$, complete or fractional, is of size at least $1$. Thus this job (complete or fractional) used to belong to a \onei{} in $\sch$ and thus $O$ satisfies \nbcon{} in $\sch'$.   Adding idle time of length $2$ to the schedule maintains the \nbcon.  
	
	Thus, steps \ref{ieps_idle_time} and \ref{ieps_idle_time2} maintains the \nbcon.
		Step \ref{med_lar_rem} maintains the \nbcon\ since every \nonei\ containing small or tiny jobs in the resulting schedule, contains jobs that were in a common \nonei\ prior to this step as there is idle time of length at least $2$ between such subsets of small or tiny jobs, and furthermore every other \nonei\ is either empty or contains at most two jobs (that are medium or large jobs and thus of length larger than $2$) so \nbcon\ is satisfied using $B\geq 2$. 
Step \ref{rearr} maintains \nbcon\ because we can permute the jobs starting in every \epsb{} in an arbitrary way as long as the set of jobs starting in this block is not changed and the \nbcon{} will continue to hold. 
	
	Thus in the schedule $\sch'$ of machine $i$ resulting at the end of step  \ref{rearr}, the following holds.
	\begin{outline}
	\1 The jobs in the schedule appear in the order of tiny or small jobs followed by medium jobs and then large jobs.
	\1 In every \epsb{} of the schedule the idle time is at the beginning of the \epsb{} or at the completion time of the job that is being processed at the start of the \epsb{} and the jobs that start in the \epsb{} continue their processing continuously, in non-decreasing order of processing time.
	\1 There is a unit length  idle time between processing a pair of jobs if there is a difference of at least $1/\eps$ time between their starting times.
	\1 The schedule satisfies the \nbcon{}.
	\end{outline}
	
	The load of every machine $i$ is at most $(1+17\eps) L_i$ and thus the resulting makespan is at most $1+17\eps$ times the makespan of $\sch$.
\end{proof}

Let $\cmm$ be the makespan of the best nice solution (of $I'$) rounded up to the next integer power of $1+\eps$.  Since we have guessed $\cm$ already, we can also guess $\cmm$ in a constant number of iterations based on the last lemma.  In what follows, we assume that the algorithm knows the value of $\cmm$ and shows that if it uses a correct upper bound on the makespan of the best nice solution then it results in a feasible schedule of makespan at most $(1+\eps)^c \cmm$ for some constant value of $c$.

\section{The mixed-integer linear program\label{sec:MILP}}
	The Mixed Integer Linear program (MILP) we exhibit below uses \textit{containers} and \textit{configurations} defined for our problem.	
	A solution of the MILP defines a choice of configurations and containers (among other information regarding the schedule). A configuration or container can be chosen multiple times and the MILP ensures that the containers satisfy the \nbcon{} (consequently the \bcon{}). 
	
	\paragraph{Containers.}
	A container $t$ stores information for a set of consecutive \epsb s, denoted by $K_t$ such that $|K_t|\leq\ieps^2$. Each $k\in K_t$ corresponds to an \epsb\ and each component is associated with a sub-vector of length $|\sm| + 4$. 	
	A sub-vector associated to the \epsb{} $k\in K_t$ contains the following components. The first $|\sm|$ components store the number of small jobs of size $l$, for every $l\in \sm$, assigned to start in the \epsb{} $k$, and it is denoted by $S_{lkt}$. The next component, $T_{kt}$, stores the lower bound on the total size of tiny jobs rounded down to an integer multiple of $\eps^2$ in the sense that  $T_{kt} \eps^2$ is the lower bound and  $T_{kt}$ is a non-negative integer. The next component, $T'_{kt}$, is a binary indicator and is $1$ if there is at least one tiny job that starts in this \epsb{}.  The next component, $D_{kt}$, stores the length of the idle time in the \epsb{} rounded down to an integer multiple of $\eps^2$ in the sense that  $D_{kt} \eps^2$ is the lower bound on the length of this interval of idle time and  $D_{kt}$ is a non-negative integer. The last component is a binary indicator denoted by $P_{kt}$. $P_{kt}$ is $1$ when the last job assigned to \epsb{} at most $k-1$ continues its processing at the beginning of \epsb{} $k$, else it is $0$.

The \textit{load} of a container $t$, $L_t$, is the total size of the jobs and idle time assigned to the container and an additional idle time of one unit at the end of the container i.e.
	\begin{align*}
		L_t = \sum_{k\in K_t}\left(D_{kt}\eps^2 + T_{kt}\eps^2 + \sum_{l\in\sm}l \cdot S_{lkt}\right) + 1.
	\end{align*}

Since each non-binary component in a sub-vector corresponding to an \epsb\ is a non-negative integer of at most $1/\eps$, there are at most $1/\eps+1$ possible values for such component and we conclude the following lemma.
	\begin{lemma}
		The number of possible containers denoted as $\cont$ is at most $\left((\ieps + 1)^{|\sm|+2}\cdot 4\right)^{\ieps^2}$.
	\end{lemma}
		
	\paragraph{Short and long containers.} Containers are further classified based on the load of the container as below. A container is \textit{short} if the load of the container is at most $\eps\cmm$, else it is \textit{long}. The load of a long container is at most $\cmm+1$.  
	
	\paragraph{Feasibility of a container.} Create a partial schedule from a container $t$. In an increasing order of $k\in K_t$ do the following. Add $D_{kt} \eps^2$ length of idle time at the beginning of the \epsb{} if $P_{kt}=0$ or the finishing time of the job being processed at the beginning of the \epsb{} if $P_{kt}=1$. Then add a virtual job of size $T'_{kt}\cdot T_{kt}\cdot \eps^2$, to account for the tiny jobs. Then add the small jobs in sequence one after the other in non-decreasing order of the size of the job. Container $t$ is feasible if at most one job continues its processing out of every \epsb{} in the partial schedule generated from the container. Let $\conts$ be the set of feasible containers. Let $\conl$ be the set of feasible long containers and $\cons$ be the set of feasible short containers. The cardinalities of these sets, namely $|\conts|$, $|\conl|$, and $|\cons|$, are upper bounded by $\cont$.

	\paragraph{Rounding load of containers.} The load of each container $t$ is rounded down to the next integer power of $1+\eps$ and the resulting rounded load is denoted by $L'_t$, i.e.,
$L'_t = \oeps^{\left\lfloor\log_{\oeps} L_t\right\rfloor}$.
	The number of distinct values of (rounded) load of long containers is at most $\log_{\oeps}\left(\frac{2}{\eps}\right) + 2$. Let $\conls$ be the set of distinct rounded loads of long containers and $|\conls |$ is upper bounded by a function of $\eps$.
	
	\paragraph{Configurations.}
	A \textit{configuration} $c$ is a vector of length $|\conls|+|\la|+4$. For $l\in\conls$, we have a component  $\alpha_{cl}$ that stores the number of long containers of rounded load $l$ assigned to that configuration. The next component $\beta_c$ stores the lower bound on the total load of short containers assigned to the configuration rounded down to an integer multiple of $\eps^2\cmm$ in the sense that  $\beta_c \eps^2 \cmm$ is the lower bound and  $\beta_c$ is a non-negative integer.  The next component is a binary indicator $\beta'_c$ that  is $1$ if at least one short container is assigned to $c$. The next $|\la|$ components store the number of large jobs of each size  assigned to a machine with this configuration, for $\mathtt{l}\in \la$ the corresponding component is denoted as $\gamma_{c\mathtt{l}}$. The next component $\delta_c$ stores the lower bound on the total size of the medium jobs rounded down to an integer multiple of $\eps^2\cmm$ in the sense that  $\delta_c \eps^2 \cmm$ is the lower bound and  $\delta_c$ is a non-negative integer.  The last component $\delta'_c$ is a binary indicator and is $1$ if at least one medium job is assigned to $c$.
	
	The \textit{load} of a configuration $c$ is denoted as $L_c$ and it is defined as
	\begin{align*}
	L_c = \sum_{l\in \conls}l\alpha_{cl} + \beta_c\eps^2\cmm + \sum_{\mathtt{l}\in\la}\mathtt{l}\gamma_{c\mathtt{l}} + \delta_c\eps^2\cmm.
	\end{align*}	
	
	\paragraph{Feasibility of a configuration.}
	A configuration is feasible if the containers assigned to the configuration are feasible and the load of the configuration is at most $\max\{ \oeps\cmm, \cmm+1\}$. In addition to that, if $\cmm\leq 1/\eps$ then the configuration will consist of exactly one long container of load at most $\cmm+1$ (and no other containers, medium jobs, or large jobs). Let $\confs$ be the set of feasible configurations.

The load of the configuration is upper bounded by $\cmm+1$ since the load of the containers is increased by unit idle time, so the load of a unique container could be larger than the last completion time of a job assigned to a machine with this container.  The next lemma follows by upper bounding the number of different values each component on a configuration can take.
	
	\begin{lemma}	
		The number of possible feasible configurations  is at most $(2/\eps)^{{|\conls|} + |\la|+4}\cdot 4$ and we denote this constant (when $\eps$ is fixed) as $\conf$.
	\end{lemma}

We summarize the last defined constant terms in Table \ref{table1}.
	
	\begin{table}[htb]
	\begin{center}
	\begin{tabular}{|c || p{25mm} || c |}
		\hline
		Notation & Description & Upper bound as a function of $\eps$\\
		\hline
		&the number & \\
		$\cont$ & of possible & $\left((\ieps + 1)^{|\sm|+2}\cdot 4\right)^{\ieps^2}$ \\ & containers  & a single-exponential function of $1/\eps$\\
		\hline
		&the number&\\
		$|\conls|$ &  of possible & $\log_{\oeps}\left(\frac{2}{\eps}\right) + 2$ \\
		& loads of large containers& a polynomial in $1/\eps$ \\
		\hline
		&the number&\\
		$\conf$ &  of possible & $ (2/\eps )^{{|\conls|} + |\la|+4}\cdot 4$ \\  & configurations  & a single-exponential function of $1/\eps$ \\
		\hline
	\end{tabular}
	\caption{Upper bounds as functions of $\eps$\label{table1}}
\end{center}
	\end{table}
	
\paragraph{Formulation of the MILP.}  We are ready to formulate the MILP.  First, we explain the set of decision variables, then we present the linear constraints.  Our MILP is a feasibility mathematical program, that is, there is no objective function, but we are only interested in finding a feasible solution (if there is such) or report that there is no feasible solution if such solution does not exist.  

\paragraph{Decision variables.}
	We have two families of decision variables. Assignment variables to assign jobs/containers to containers/configurations, and counters to count the number of times a container/configuration is chosen.  There are three types of assignment variables and two types of counters.

	\paragraph{Assignment of containers to configurations.}
	The variable $z_{ct}$ indicates the number of times a container $t$ is assigned to configuration $c$. The total number of assignment variables of containers to configurations is 
$\cont\cdot\conf$,	which is a constant when $\eps$ is fixed. They are forced to be integral.
	
	\paragraph{Assignment of tiny jobs to containers.}
	The variable $y_{jkt}$ is $1$ when tiny job $j$ is assigned to start in the $k^{th}$ \epsb{} in container $t$. The total number of assignment variables of tiny jobs to containers is at most 
$n\cdot \ieps^2 \cdot \cont$,	which is upper bounded by a polynomial in the input encoding length. These variables are allowed to be fractional.
	
	\paragraph{Assignment of medium jobs to configurations.}
	The variable $x_{jc}$ is $1$ when medium job $j$ is assigned to configuration $c$. The total number of assignment variables of medium jobs to configurations is at most 
$n\cdot\conf$,	which is upper bounded by a polynomial in the input encoding length. These variables are allowed to be fractional.

	\paragraph{Container counters.}
	The variable $w_t$ represents the number of times the container $t$ is chosen. The number of container counters is
$\cont$ which is a constant when $\eps$ is fixed. These variables are forced to be integral.
	
	\paragraph{Configuration counters.}
	The variable $v_c$ represents the number of times the configuration $c\in C$ is chosen. The number of configuration counters is  
$\conf$	 which is a constant when $\eps$ is fixed. These variables are forced to be integral.
	
We conclude that the MILP will have $ \cont\cdot \conf + \cont + \conf$ integer decision variables and at most $n \cdot (\ieps^2 \cdot \cont + \conf)$ fractional decision variables.	
	
	\paragraph{The constraints of the MILP.}  Next, we present the constraints of the MILP together with their intuitive explanation.  Note that this intuition is described only to assist us in formulating the MILP, but the MILP's relevance and correctness are not based on this intuition and they will be derived below.
	
 Assigning tiny jobs to containers
		\begin{align}
			\sum_{t\in \conts}\sum_{k\in K_t}y_{jkt} &= 1,\forall j\in\tj\label{every_tiny_ass}\\
			\sum_{j\in \tj}p_jy_{jkt}&\leq w_tT'_{kt}(T_{kt}+1)\eps^2,\ \forall k\in K_t, \forall t\in \conts\label{tot_tiny}
		\end{align}
	
		There is enough positions for all the small jobs of each size in containers
		\begin{align}
			\sum_{t\in \conts}\sum_{k\in K_t}S_{lkt}w_t = |\sj_l|, \forall l\in\sm\label{small_nos}
		\end{align}
	
		Assigning medium jobs to configurations
		\begin{align}
			\sum_{c\in\confs} x_{jc} &= 1, \forall j\in\mj\label{every_med_ass}\\
			\sum_{j\in \mj}x_{jc}p_j &\leq v_c\delta'_c(\delta_c+1)\eps^2\cmm,\ \forall c\in\confs\label{tot_med}
		\end{align}
	
		There is enough positions for all the large jobs of each size in configurations
		\begin{align}
			\sum_{c\in\confs} \gamma_{cl}v_c = |\lj_l|,\ \forall l\in\la\label{lar_nos}
		\end{align}
	
		The number of times a container is chosen is equal to the sum of the number of times that container is assigned to all the configurations. 
		\begin{align}
			\sum_{c\in\confs}z_{ct} &= w_t,\ \forall t\in\conts\label{cont_nos}
		\end{align}

	 If a configuration is not chosen, then there are no containers assigned to that configuration. Otherwise, any number of containers (and without loss of generality this is at most $n$) can be assigned to copies of that configuration
		\begin{align}
			z_{ct} &\leq n v_c,\ \forall t\in\conts, c\in\confs \label{conf_nos}
		\end{align}
	
	 Assigning short containers to configurations
		\begin{align}
			\sum_{t\in\cons} L_tz_{ct}\leq v_c\beta'_c(\beta_c+1)\eps^2\cmm,\ \forall c\in\confs \label{tot_cons}
		\end{align}
	
	 There is enough positions for all the long containers of each rounded load in the configurations
		\begin{align}
			\sum_{t:L'_t=l}w_t &= \sum_{c\in\confs} v_c\alpha_{cl},\ \forall l\in\conls\label{conl_nos}
		\end{align}
	
	We enforce \nbcon{} by enforcing the following average constraints for containers
		\begin{align}
			\sum_{k=i}^{i+\ieps} \left(w_t\sum_{l\in\sm}S_{lkt} + \sum_{j\in\tj}y_{jkt}\right)+ w_tP_{it}\leq w_t \cdot B,\ \ \ \  i=1,2,\ldots,(|K_t|-(\ieps+1)), \forall t\in \conts\label{nbcon}
		\end{align}
	
	 Number of configurations chosen equals the number of machines
		\begin{align}
			\sum_{c\in\confs}v_c = |M|\label{conf_mach}
		\end{align}
	
	 Non-negativity constraints and integrality constraints
		\begin{align}
			0 \leq y_{jkt}&\leq 1,\ \forall j\in\tj,k\in K_t,t\in\conts\\
			0 \leq x_{jc}&\leq 1,\ \forall j\in\mj,c\in\confs\\
			w_t&\in\mathbb{Z}_+,\ \forall t\in\conts\\
			v_c&\in\mathbb{Z}_+,\ \forall c\in\confs\\
			z_{ct}&\in\mathbb{Z}_+,\ \forall c\in\confs,\forall t\in\conts
		\end{align}

	A MILP solution will give $(\textbf{v},\textbf{w},\textbf{x},\textbf{y},\textbf{z})$.  We conclude this section by proving that if there is a feasible  nice schedule (for the current guessed value of the approximated makespan of the best nice solution), then the MILP has a feasible solution.

	\begin{theorem}\label{mil_obj_val}
		If there exists a feasible nice schedule to \bcs\ with makespan at most $\cmm$, then the MILP has a feasible solution.
	\end{theorem}
	\begin{proof}
		Assume that there exists a nearly optimal nice schedule \opt{} for the rounded instance $I'$. Let the makespan of \opt{} be $\copt \leq \cmm$. Let the schedule for machine $i$ in this solution be $\opt_i$.		
		We generate a feasible MILP solution from \opt{} along with a set of feasible configurations.  For all machine $i$, we first construct a configuration $C_i$ corresponding to this machine from $\opt_i$ as below.  
		
\paragraph{Defining containers.}	 Consider the prefix of the schedule of this machine consisting only of the time period where all the tiny and small jobs assigned to that machine are scheduled (together with its idle time in between those jobs). Partition the job set assigned to this prefix into containers. Two jobs are assigned to the same container only if there is no unit length idle time between their starting times.  By definition, $\opt_i $ has a unit idle time after every container and we add it to the container.  If the makespan is at most $1/\eps$, the entire schedule of $i$ is one container and we do not partition it further.  Since \opt\ is nice, the maximum load of a container is at most $2/\eps+1$.   Let $T_i$ be the set of containers generated from $\opt_i$.
			
\paragraph{partitioning each container into \epsb s.} Partition each container $t\in T_i$ into \epsb s and ignore the \epsb s that belong to the \epsb s of the unit idle time at the end of the container and the \epsb s beyond the $1/\eps^2$ blocks within the container (if there are such). 
For each \epsb{} $k$ of $t$ we define the following components of the sub-vector of $t$ corresponding to $k$.  First, let $D_{kt}$ be so that $\eps^2 D_{kt}$ is the idle time in $k$ rounded down to the next integer multiple of $\eps^2$. Second, let $	T_{kt} = \left\lfloor\frac{\sum p_j}{\eps^2}\right\rfloor$, 
				where the summation is over the set of tiny jobs that start in $k$ so $T_{kt} \cdot \eps^2$ is  the total size of tiny jobs assigned to $k$ of $t$ rounded down to the next integer multiple of $\eps^2$. If there are no tiny jobs assigned to start in $k$ then let $T'_{kt}=0$ else $T'_{kt}=1$.				
					Third, for each distinct size of small jobs $l\in\sm$ let $S_{lkt}$ be the number of small jobs of size $l$ that start in $k$ of $t$. 
Last, set $P_{kt}$ to be $1$ if a job starts in an \epsb{} at most $k-1$ and continues its processing also in the \epsb{} $k$, else it is $0$.  Based on these sub-vectors, we define the container $t$, and using this vector we define the load of the container.
			
\paragraph{Defining the configuration $C_i$ corresponding to machine $i$.}  For each $l$, we count the number of long containers with load rounded to integer powers of $1+\eps$, $l$ assigned to $i$. Denote it by $\alpha_{il}$. We compute the total load of all short containers assigned to $i$ and we round this value down to the next integer multiple of $\eps^2\cmm$. The resulting rounded value is  $\beta_i \eps^2 \cmm$. If there are no short containers assigned to $i$ then let $\beta'_i = 0$, else $\beta'_i =1$.			
			 For each $l\in\la$, count the number of large jobs of size $l$ assigned to $i$ and denote it by $\gamma_{il}$.			
			 Calculate the total size of medium jobs assigned to $i$ and round down the resulting value to the next integer multiple of $\eps^2\cmm$. Let $\delta_i$ be such that this outcome is $\delta_i \eps^2 \cmm$. If there are no medium jobs assigned to $i$ then let $\delta'_i=0$ else $\delta'_i=1$.  This defines all components of a configuration that we denote by $C_i$.

\paragraph{Feasibility of containers and configurations.} Let $T$ be the multiset of all containers, and let $C$ be the multiset of all configurations we have defined (over all machines).	 Next we check the feasibility of the containers generated. Create partial schedules from each container. Since a container was created by rounding down the total idle time and the total size of tiny jobs in \epsb s with respect to  \opt{}, the number of jobs that continue their processing out of the corresponding \epsb{} does not increase (compared to \opt{}). By the facts that at most one job continued its processing out of every \epsb{} in \opt{}, and  \opt{} is a nice solution, we conclude that the partial schedule generated from every container will also have at most one job whose processing continues out of every \epsb{}. Thus, the containers generated are feasible. Consequently we get that the configurations in $C$ are also feasible, since the containers are feasible and the load of each configuration is not increased while the configuration was generated from the schedule (with the rounding down involved) and if $\cmm\leq 1/\eps$ then the configuration consists of one long container and no other jobs or containers. 
		
\paragraph{Defining the MILP solution.}		Two configurations in $C$ are identical if every parameter of the two configurations are identical. Let $v_c$ for every feasible configuration $c$ be the number of identical copies of configuration $c$ in the multiset $C$ (perhaps $0$ if $c\notin C$). Thus now we have $\textbf{v}^{\opt{}}$. Two containers are identical if every parameter of the two containers are identical. Let $w_t$ for all feasible containers $t$ be the number of identical copies of container $t$ in the multiset $T$. Thus now we have $\textbf{w}^{\opt{}}$.  Identify the configurations to which each container is assigned to and obtain $\textbf{z}^{\opt}$. Identify the \epsb{} and container to which every tiny job is assigned to start to obtain $\textbf{y}^{\opt}$. Identify the configurations to which each medium job is assigned to to obtain $\textbf{x}^{\opt}$. 
		
\paragraph{Proving that the MILP solution is feasible.}	In order to prove the theorem it suffices to show that the generated MILP solution $(\textbf{v}^\opt,\textbf{w}^\opt,\textbf{x}^\opt,\textbf{y}^\opt,\textbf{z}^\opt)$ is feasible. 	
		From the definition of the MILP solution $(\textbf{v}^\opt,\textbf{w}^\opt,\textbf{x}^\opt,\textbf{y}^\opt,\textbf{z}^\opt)$, all of them are non-negative integers and $\textbf{x}^\opt$ and $\textbf{y}^\opt$ are binary vectors. 
		Since each  tiny job is assigned to start only in one \epsb{} of only one container of one machine, and similarly since each medium job is scheduled only to one machine, constraints  \eqref{every_tiny_ass} and \eqref{every_med_ass}  are satisfied.	 Consider a container $t\in T$ with only one copy of $t$ and an \epsb{} $k$ of the container. Thus we have
		\begin{align*}
			T_{kt} &= \left\lfloor\sum_{j\in\tj}\frac{p_jy_{jkt}}{\eps^2}\right\rfloor\\
			(T_{kt}+1)\eps^2 &\geq \sum_{j\in\tj}p_jy_{jkt}\\
		\intertext{Considering all copies of $t$ for an arbitrary container we get}
			w_t (T_{kt}+1)\eps^2 &\geq \sum_{j\in\tj}p_jy_{jkt}.
		\end{align*}
		Thus constraint \eqref{tot_tiny} is satisfied as it is trivially satisfied if no tiny jobs are assigned to \epsb\ $k$ of $t$.	Constraint \eqref{small_nos} is satisfied trivially from the definition of the variables. Consider a configuration $c\in C$ with only one copy of $c$. Thus we have
		\begin{align*}
			\delta_c &= \left\lfloor\sum_{j\in\mj}\frac{x_{jc}p_j}{\eps^2\cmm}\right\rfloor\\
			(\delta_c+1)\eps^2\cmm &\geq\sum_{j\in\mj}x_{jc}p_j\\
		\intertext{Considering all copies of $c$ for an arbitrary configuration we get}
			v_c (\delta_c+1)\eps^2\cmm &\geq\sum_{j\in\mj}x_{jc}p_j.
		\end{align*}
		Thus constraint \eqref{tot_med} is satisfied as it is trivially satisfied if no medium jobs are assigned to the configuration. If a container was generated from the schedule of a machine, then that container will be placed in only one configuration (of the same machine). Similarly, a large job will also be placed in only one configuration. Thus constraints \eqref{lar_nos},  \eqref{cont_nos},   \eqref{conf_nos}, and \eqref{conl_nos} are satisfied.  Recall that $\cons$ is the set of short containers. Consider a configuration $c$ with only one copy of $c$,  we have
		\begin{align*}
			\beta_c &= \left\lfloor\sum_{t\in\cons}\frac{z_{ct}L_c}{\eps^2\cmm}\right\rfloor\\
			(\beta_c+1)\eps^2\cmm &\geq \sum_{t\in\cons}z_{ct}L_c\\
		\intertext{Considering all the copies of $c$ for an arbitrary configuration we get}
			v_c(\beta_c+1)\eps^2\cmm &\geq \sum_{t\in\cons}z_{ct}L_c.
		\end{align*}
		Thus constraint \eqref{tot_cons} is satisfied as it is trivially satisfied if no short container is assigned to $c$. Constraint \eqref{nbcon} is satisfied since \opt{} was a nice schedule and satisfied the \nbcon. 
 Constraint \eqref{conf_mach} is satisfied trivially since a configuration was generated from the schedule of each machine.
		Hence, the solution generated is feasible as required.
	\end{proof}

\section{Rounding a MILP solution into a feasible schedule\label{sec:round-MILP}}
	Next we state the main claim regarding the performance guarantee of the Best-Fit schedule for scheduling with cardinality constraint \cite[Lemma~2]{LJ+16} which we use later in our algorithm (in their settings, $c_k\in\mathbb{Z}_+$ is the capacity bound on machine $k$ but here the use will be different).
	\begin{lemma}\label{best_fit}
		Assume that $c_k$ is a non-negative integer for all $k$. If there is a feasible solution for the following linear system 
		\begin{align*}
			\sum_{j=1}^{n} p_jy_{jk} \leq t_k&,\ 1\leq k\leq m\\
			\sum_{j=1}^{n} y_{jk} = c_k&,\ 1\leq k\leq m\\
			\sum_{k=1}^{m} y_{jk} = 1&,\ 1\leq j\leq n\\
			0\leq y_{jk} \leq 1&, \ 1\leq j\leq n,\ \ 1\leq k\leq m
		\end{align*}
	then an integer solution satisfying:
	\begin{align*}
		\sum_{j=1}^{n} p_jy_{jk} \leq t_k + p_{max}&,\ 1\leq k\leq m\\
		\sum_{j=1}^{n} y_{jk} = c_k&,\ 1\leq k\leq m\\
		\sum_{k=1}^{m} y_{jk} = 1&,\ 1\leq j\leq n\\
		y_{jk} \in \{0,1\}&, \ 1\leq j\leq n,\ \ 1\leq k\leq m
	\end{align*}
	could be obtained in $O(n\log n)$ time, where $p_{\max}=\max_{j}\{p_j\}$. 
	\end{lemma}

We are now ready to show the algorithm for rounding the MILP solution and its analysis.
	
	\begin{theorem}
		If there exists a solution to the MILP for the rounded instance $I'$ with the given guessed value $\cmm$, then there exists a feasible schedule to instance $I'$ of cost at most $(1+10\eps)\cmm$ that can be obtained in polynomial time.
	\end{theorem}
	\begin{proof}
		Let the MILP solution be denoted by \sol{} and is given as $(\textbf{v},\textbf{w},\textbf{x},\textbf{y},\textbf{z})$. Let $\confs$ be the set of feasible configurations and let $\conts$ be the set of feasible containers corresponding to this value of $\cmm$. We create a schedule from \sol{}, $\conts$, and $\confs$.
		
		\paragraph{Assigning configurations to machines.} Assign the configurations with $v_c>0$ in $\confs$ to the machines in a way that every machine is assigned exactly one configuration and for every $c$, the number of machines whose assigned configuration is $c$ will be exactly $v_c$. By constraint \eqref{conf_mach}, $\sum_{c\in\confs}v_c = |M|$, this is indeed possible. 
		
		\paragraph{Assigning large jobs to machines.} We assign large jobs to each machine based on the positions available for each distinct size of large jobs in the configuration assigned to that machine.  That is, for each size $l$ of large jobs, if the configuration assigned to $i$ has a given number of jobs of size $l$ assigned to it, then $i$ will be assigned exactly the same number of large jobs of that size. From constraint \eqref{lar_nos}, we assign all the large jobs to the machines.
		
		\paragraph{Assigning medium jobs.} For each machine we have the lower bound on the total size of the medium jobs that can be assigned on that machine based on the configuration assigned to the machine.  This lower bound is given by $\delta'_c\delta_c\eps^2\cmm$. Since $\delta_c$ was rounded down we increase the space available for the medium jobs by $\eps^2\cmm$. By constraint \eqref{tot_med}, we have enough space to assign all the medium jobs fractionally. We increase the total space available on each machine for medium jobs by another additive term of $\eps\cmm$. We assign the medium jobs in some order to machines in some order of the machines. In this sequential assignment, we move to the next machine if the medium job about to be assigned to the current machine will increase the total size of the medium jobs on the current machine beyond the space available on the current machine. Thus the total increase in makespan due to this step is at most $\eps^2\cmm+ \eps\cmm \leq 2\eps\cmm$.  This assignment of medium jobs will assign all these jobs due to the following reasons.  Assume by contradiction that some medium jobs are left unassigned by this process.  This means that whenever the procedure decided that some machine should stop being the current machine, the total size of medium jobs assigned to it will be at least $\delta'_c(\delta_c+1)\eps^2\cmm$ (this holds as the size of a medium job is not larger than $\eps \cmm$).  Since every machine stops being the current machine at some point, we get a contradiction to constraint  \eqref{tot_med}.  So indeed this is a feasible assignment of medium jobs.
		
		\paragraph{Assigning long containers to machines.} Long containers are assigned to the machines based on the positions available for each distinct rounded load of the long containers defined in the configuration assigned to the machine. The long containers of a common rounded load are assigned in an arbitrary order and from constraint \eqref{conl_nos} there are enough positions to assign all the long containers to the machines. 
		
		\paragraph{Assigning short containers to machines.} Based on the configurations of machines, we know a lower bound on the total load available for short containers on each machine. We increase this space by $\eps^2\cmm$ (the motivation for this increase is that the load bound of the short containers were rounded down). From constraint \eqref{tot_cons} we have enough space to assign fractionally all the short containers to the machines. We increase the total space available for short containers on each machine by another additive term of $\eps\cmm$. Similarly to the assignment of medium jobs, we apply the following iterative assignment rule. We assign the short containers listed in some order to machines where machines are ordered in some order. We move to the next machine if the short container about to be assigned to the current machine will increase the total load of the short containers on the current machine beyond the space available to short containers on this machine. The increase in makespan is at most $2\eps\cmm$.   Once again, assume by contradiction that some short containers are left unassigned by this process.  This means that whenever the procedure decided that some machine should stop being the current machine, the total load of short containers assigned to it is at least $\beta'_c(\beta_c+1)\eps^2\cmm$ (this holds as the load of a short container is not larger than $\eps \cmm$).  Since every machine stops being the current machine at some point, we get a contradiction to constraint  \eqref{tot_cons}.  So indeed this is a feasible assignment of short containers.
		
		\paragraph{Ordering the containers and medium and large jobs assigned to each machine.} The containers assigned to the machines are sorted in non-decreasing order of load and are placed sequentially one after the other (starting at time $0$). Recall that all containers contain a unit idle time at the end, and containers start and end at integer multiples of $\eps$. Next, starting at the end-point of the unit idle time of the last container assigned to the machine or at time $0$ if there is no such container, we place the medium and large jobs. These jobs are assigned to the machine and placed in non-decreasing order of size sequentially without idle time between consecutive such jobs. This does not increase the load of any machine. That is, at this point in time we have a schedule of containers and medium and large jobs with makespan at most $(1+4\eps) \cmm$.
		
\paragraph{Assigning small jobs to machines.} We consider the small jobs and positions for each distinct size of small jobs on each machine based on the containers assigned to that machine. The total positions available for each distinct size of small jobs on a machine is obtained by summing up all the positions available for that distinct size of small job over all \epsb s over all the containers assigned to the machine. From constraint \eqref{small_nos}, we have enough positions to assign all the small jobs of each size of small jobs to the machines and the \epsb s of the machines.		
		
\paragraph{Assigning the tiny jobs.}
We are left with the assignment of the tiny jobs that is the most delicate part of our proof.	In a nutshell we would like to use the MILP solution in order to define a fractional assignment of tiny jobs to \epsb s that maintain the total size of these jobs (to be similar to those defined by the containers), and the capacity, that is, the fractional number of those jobs in each \epsb, and then to use the analysis of Best-Fit in Lemma \ref{best_fit}.  However, the main difficulty in this approach is that the capacity bounds used in that lemma should be integers whereas the fractional assignment of tiny jobs have fractional cardinalities and we should first round those cardinality while enforcing the \nbcon.

We define \uepsb s. The \uepsb s are common to all machines and every \epsb\ on some machine appear as one of those \uepsb.  We treat such \uepsb\ as an ordered pair where the second component of the pair is the index of the machine and the first component is an index along the time horizon. Let $\widetilde{U}$ be the set of values of the first components of the \uepsb s, and let $U=\widetilde{U}\times M$ denote the set of all \uepsb s. Each \epsb{} $k$ in a container $t$ assigned to machine $i$ corresponds to a $u\in U$ with second component equal to $i$.  We  define $T_{u}$, $S_{lu}$, $P_{u}$, and $y_{ju}$ using $T_{u} = T_{kt}$, $S_{lu} = S_{lkt}$, $P_{u} = P_{kt}$ and $y_{ju} = y_{jkt}$ where $u$ corresponds to an \epsb\ $k$ of container $t$ assigned to machine $i$. 
		
 The MILP solution implies that $y_{jkt}$ amount of tiny job $j$ is assigned to start in \epsb{} $k$ of container $t$ in total on $w_t$ copies of container $t$. We define a temporary fractional assignment of tiny jobs that assigns $\frac{y_{jkt}}{w_t}$ fraction of job $j$ to start in \epsb\ $k$ of each occurrence of container $t$. The number of tiny jobs that starts in each $u\in U$ in this temporary fractional solution is $\Omega_{u} = \sum_{j\in \tj}y_{jkt}/w_t$  where $u$ corresponds to the \epsb{} $k$ of container $t$ assigned to machine $i$. 
 
 We will apply below the Best-Fit analysis to get integer assignments of tiny jobs. In order to apply it, for each $u\in U$, we need integer capacity bound. Denote it by $\Omega'_{u}$.  We will have $\Omega'_{u} = 0$ when $\Omega_{u} = 0$. The remaining $\Omega'$ values are obtained by using the \ref{LP} below. The first constraint is used to round down or up the $\Omega$ values (in order to not increase by too much the total size of tiny jobs in each \epsb\ in a new fractional assignment that we will construct). The second constraint enforces the \nbcon{} while rounding the $\Omega$ values. The third constraint ensures that there is enough positions to assign all the tiny jobs of the instance with the integer capacities.
		\begin{align*}\tag{LP}\label{LP}
			\lfloor \Omega_{u}\rfloor \leq\ &\Omega'_{u} \leq \lfloor \Omega_{u}\rfloor +1,\ \forall u\in U\\
			\sum_{u=(k,i)}^{(k+\ieps,i)} &\Omega'_{u} \leq w_t \left(B - \sum_{u=(k,i)}^{(k+\ieps,i)} \sum_{l\in\sm}S_{lu} - P_{(k,i)}\right),\ \forall (k,i): k \mbox{ and } k+1/\eps \\ &  \mbox{ are indexes of \epsb s on machine $i$ that belong to a common container $t$}, \forall t\\ 
			\sum_{u\in U}&\Omega'_{u} = |\tj|
		\end{align*}
		\ref{LP} has a feasible solution since $\Omega_{u}$ for all $u$ is a feasible solution. The constraint matrix of the given \ref{LP} satisfies that in every row the $1$'s appear consecutively for sorted list of the columns where we order the columns in a non-decreasing order of the machines (second component of $u$) and break ties in an increasing order of $k$ (first component of $u$).	That is, the constraint matrix satisfies the consecutive $1$'s property, and so it is a totally unimodular matrix. Since the constraint matrix of \ref{LP} is  totally unimodular, the right hand side is integral, and there is a fractional feasible solution for \ref{LP}, we conclude that  there exists an integer solution to $\Omega'_{u}$ that can be found in polynomial time by a basis crashing algorithm.  We get integer capacity bound for tiny jobs, for each \uepsb{} $u$. Denote by $\underline{U}\subseteq U$ the subset of indexed \uepsb s that has an integer capacity bound that is rounded down with respect to the $\Omega$ values and also those indexed \uepsb s that already had an integer capacity bound in $\Omega$. Let $\overline{U} = U\backslash \underline{U}$ be the subset of indexed \uepsb s that have integer capacity bounds that are rounded up and strictly larger than the $\Omega$ values. 
		
		In order to create a feasible fractional assignment of the tiny jobs with respect to the integer capacity bounds, we perform the below operations to alter the $y$ values. Let the altered $y$ values be denoted by $y'$. For each $u\in\underline{U}$, perform the following transformation. Let 
		\[
		y'_{ju} = \frac{y_{ju}}{\Omega_u}\cdot \Omega'_u,
		\]
		Thus for $u$,
		\[
		\sum_{j\in \tj} \frac{y'_{ju}}{w_t} = \sum_{j\in \tj} \frac{y_{ju}}{\Omega_uw_t}\cdot \Omega'_u = \Omega'_u, 
		\]   
		the $y'$ values satisfy the integer capacity bound. This transformation may leave some fractions of jobs unassigned from $u\in\underline{U}$. Merge the fractions of the common jobs to get the set of fractions $Y$ for the unassigned jobs. For each $u\in\overline{U}$ perform the following operations. Let 
		\[
			y'_{ju} = y_{ju}.
		\]
		In addition to these fractions, assign (fractionally) fractions from $Y$ to $u$ greedily till the total fraction of jobs in $u$ is exactly $\Omega'_u$, splitting a fraction of the same job when necessary. There is enough place to assign all the fractions in $Y$ since
		\[
			\sum_{u\in\underline{U}}\Omega'_u + \sum_{u\in\overline{U}}\Omega'_u = \sum_{u\in U}\Omega'_u = |\tj| = \sum_{j\in \tj}\sum_{t\in \conts}\sum_{k\in K_t} y_{jkt},
		\]
		from \ref{LP} and MILP. This transformation can increase the total size allocated to tiny jobs in each \uepsb{} by at most $\eps^2$.  Furthermore, we apply Lemma \ref{best_fit} to assign integrally the tiny jobs to the \uepsb\ and this increases the total size of tiny jobs in each \uepsb\ by another additive term of $\eps^2$. Thus, the new load of the schedule of a machine assigned configuration $c$ is increased due to the assignment of tiny jobs by at most $$2(1+4\eps) \cdot \frac{\cmm}{\eps} \cdot\eps^2 \leq 6 \eps\cmm .$$

 Thus the final makespan of the schedule is at most $(1+10 \eps)\cmm$ and all the operations are done in polynomial time.
	\end{proof}

We conclude by the sequence of lemmas and theorems that our problem admits an EPTAS and thus the following corollary holds.
\begin{corollary}
Problem \bcs\ admits an EPTAS.
\end{corollary}


\begin{thebibliography}{10}

\bibitem{Alon98}
N.~Alon, Y.~Azar, G.~J. Woeginger, and T.~Yadid.
\newblock Approximation schemes for scheduling on parallel machines.
\newblock {\em Journal of Scheduling}, 1(1):55--66, 1998.

\bibitem{Benmansour2014}
R.~Benmansour, O.~Braun, and A.~Artiba.
\newblock {On the single-processor scheduling problem with time restrictions}.
\newblock {\em International Conference on Control, Decision and Information
  Technologies, CoDIT 2014}, pages 242--245, 2014.

\bibitem{Benmansour2019}
R.~Benmansour, O.~Braun, and S.~Hanafi.
\newblock {The single-processor scheduling problem with time restrictions:
  complexity and related problems}.
\newblock {\em Journal of Scheduling}, 22(4):465--471, 2019.

\bibitem{Braun2014}
O.~Braun, F.~Chung, and R.~Graham.
\newblock {Single-processor scheduling with time restrictions}.
\newblock {\em Journal of Scheduling}, 17(4):399--403, 2014.

\bibitem{Braun2016}
O.~Braun, F.~Chung, and R.~Graham.
\newblock {Worst-case analysis of the LPT algorithm for single processor
  scheduling with time restrictions}.
\newblock {\em OR Spectrum}, 38(2):531--540, 2016.

\bibitem{Cesati97}
M.~Cesati and L.~Trevisan.
\newblock On the efficiency of polynomial time approximation schemes.
\newblock {\em Information Processing Letters}, 64(4):165 -- 171, 1997.

\bibitem{LJ+16}
L.~Chen, K.~Jansen, W.~Luo, and G.~Zhang.
\newblock An efficient {PTAS} for parallel machine scheduling with capacity
  constraints.
\newblock In {\em International Conference on Combinatorial Optimization and
  Applications {(COCOA)}}, pages 608--623, 2016.

\bibitem{DIMM06}
M.~Dell’Amico, M.~Iori, S.~Martello, and M.~Monaci.
\newblock Lower bounds and heuristic algorithms for the $k_i$-partitioning
  problem.
\newblock {\em European Journal of Operational Research}, 171(3):725--742,
  2006.

\bibitem{DM01}
M.~Dell'Amico and S.~Martello.
\newblock Bounds for the cardinality constrained {$P||C_{max}$} problem.
\newblock {\em Journal of Scheduling}, 4(3):123--138, 2001.

\bibitem{Downey99}
R.~Downey and M.~Fellows.
\newblock {\em Parameterized Complexity}.
\newblock Springer, New York, NY, 1999.

\bibitem{epstein2014}
L.~Epstein and A.~Levin.
\newblock An efficient polynomial time approximation scheme for load balancing
  on uniformly related machines.
\newblock {\em Mathematical Programming}, 147(1-2):1--23, 2014.

\bibitem{EL14c}
L.~Epstein and A.~Levin.
\newblock Minimum total weighted completion time: Faster approximation schemes.
\newblock {\em CoRR}, abs/1404.1059, 2014.

\bibitem{Flum2006}
J.~Flum and M.~Grohe.
\newblock {\em Parameterized Complexity Theory}.
\newblock Springer-Verlag Berlin Heidelberg, 2006.

\bibitem{He03}
Y.~He, Z.~Tan, J.~Zhu, and E.~Yao.
\newblock $k$-partitioning problems for maximizing the minimum load.
\newblock {\em Computers \& Mathematics with Applications},
  46(10-11):1671--1681, 2003.

\bibitem{hochbaum97}
D.~S. Hochbaum, editor.
\newblock {\em Approximation algorithms for {NP}-hard problems}.
\newblock PWS Pub. Co, Boston, 1997.

\bibitem{Hochbaum87}
D.~S. Hochbaum and D.~B. Shmoys.
\newblock Using dual approximation algorithms for scheduling problems
  theoretical and practical results.
\newblock {\em Journal of the ACM}, 34(1):144--162, 1987.

\bibitem{Hochbaum88}
D.~S. Hochbaum and D.~B. Shmoys.
\newblock A polynomial approximation scheme for scheduling on uniform
  processors: Using the dual approximation approach.
\newblock {\em SIAM Journal on Computing}, 17(3):539--551, 1988.

\bibitem{Jansen10}
K.~Jansen.
\newblock An {EPTAS} for scheduling jobs on uniform processors: Using an {MILP}
  relaxation with a constant number of integral variables.
\newblock {\em {SIAM} Journal on Discrete Mathematics}, 24(2):457--485, 2010.

\bibitem{JKMR19}
K.~Jansen, K.~Klein, M.~Maack, and M.~Rau.
\newblock Empowering the configuration-{IP} - new {PTAS} results for scheduling
  with setups times.
\newblock In {\em {10th conference on Innovations in Theoretical Computer
  Science {(ITCS)}}}, pages 44:1--44:19, 2019.

\bibitem{JKV16}
K.~Jansen, K.-M. Klein, and J.~Verschae.
\newblock Closing the gap for makespan scheduling via sparsification
  techniques.
\newblock {\em Mathematics of Operations Research}, 45(4):1371--1392, 2020.

\bibitem{Jansen2019}
K.~Jansen and M.~Maack.
\newblock An {EPTAS} for scheduling on unrelated machines of few different
  types.
\newblock {\em Algorithmica}, 81(10):4134--4164, 2019.

\bibitem{kannan1983}
R.~Kannan.
\newblock Improved algorithms for integer programming and related lattice
  problems.
\newblock In {\em 15th annual ACM symposium on Theory of computing {(STOC)}},
  pages 193--206, 1983.

\bibitem{Ka21}
Y.~Kawase, K.~Kimura, K.~Makino, and H.~Sumita.
\newblock Optimal matroid partitioning problems.
\newblock {\em Algorithmica}, 83(6):1653--1676, 2021.

\bibitem{KK13}
H.~Kellerer and V.~Kotov.
\newblock A 3/2-approximation algorithm for $k_i$-partitioning.
\newblock {\em Operations Research Letters}, 39(5):359--362, 2011.

\bibitem{Kones2019}
I.~Kones and A.~Levin.
\newblock A unified framework for designing {EPTAS} for load balancing on
  parallel machines.
\newblock {\em Algorithmica}, 81(7):3025--3046, 2019.

\bibitem{lenstra1983}
H.~W. Lenstra~Jr.
\newblock Integer programming with a fixed number of variables.
\newblock {\em Mathematics of operations research}, 8(4):538--548, 1983.

\bibitem{Zhang2018}
A.~Zhang, Y.~Chen, L.~Chen, and G.~Chen.
\newblock {On the NP-hardness of scheduling with time restrictions}.
\newblock {\em Discrete Optimization}, 28:54--62, 2018.

\bibitem{Zhang2017}
A.~Zhang, F.~Ye, Y.~Chen, and G.~Chen.
\newblock {Better permutations for the single-processor scheduling with time
  restrictions}.
\newblock {\em Optimization Letters}, 11(4):715--724, 2017.

\end{thebibliography}
\end{document}